\documentclass[prx,aps,floatfix,amsmath,superscriptaddress,twocolumn,longbibliography,nofootinbib]{revtex4-2}
\pdfoutput=1

\usepackage{amssymb,enumerate}
\usepackage{graphicx}
\usepackage{amsmath}
\usepackage{amsthm}
\usepackage{color}
\usepackage{dsfont}
\usepackage{hyperref}
\usepackage{boldline}
\usepackage{cleveref}
\usepackage{tikz}
    \usetikzlibrary{math}
    \usetikzlibrary{decorations.pathreplacing}
    \usetikzlibrary{quantikz}
\usepackage{adjustbox}
\usepackage{xcolor}
\hypersetup{
    colorlinks,
    linkcolor={red!50!black},
    citecolor={blue!50!black},
    urlcolor={blue!80!black}
}
\usepackage{xfrac}
\usepackage{empheq}
\usepackage{enumitem}
\usepackage{tabularray}
\usepackage{graphicx,graphics,epsfig,times,bm,bbm,amssymb,amsfonts,mathrsfs}
\usepackage[normalem]{ulem}
\usepackage{subfigure}
\usepackage{braket}
\usepackage{natbib}
\usepackage{chngcntr}

\newcommand{\Sn}{\mathbb{S}_n}

\newcommand{\R}{\mathbb{R}}
\newcommand{\C}{\mathbb{C}}

\newcommand{\Tr}{\operatorname{Tr}}

\newcommand{\<}{\langle}
\renewcommand{\>}{\rangle}
\renewcommand{\H}{\mathcal{H}}

\newcommand{\HTC}{{H}_{\text{TC}}}

\newcommand{\jmin}{j_\text{min}}

\newcommand{\qj}{_{q,j}}

\newcommand{\vac}{\ket{0}_{\text{osc}}}

\newcommand{\gTC}{g_{\text{TC}}}

\newcommand{\1}{\mathbb{I}}

\newcommand{\eq}{\,=\,}
\newcommand{\s}{\text{Span}}

\newcommand{\ipo}[3]{\left\langle #1 \middle| #2 \middle| #3 \right\rangle}

\newcommand{\bes}{\begin{subequations}}
\newcommand{\ees}{\end{subequations}}
\newcommand{\bea}{\begin{eqnarray}}
\newcommand{\eea}{\end{eqnarray}}
\newcommand{\be}{\begin{equation}}
\newcommand{\ee}{\end{equation}}

\def\al{\alpha}

\def\>{\rangle}
\def\<{\langle}
\def\Tr{\textrm{Tr}}

\newcommand{\ignore}[1]{}

\newtheorem{thm}{Theorem}

\newtheorem{theorem}{Theorem}

\newtheorem{proposition}[theorem]{Proposition}

\crefname{tab}{Table}{Tables}
\crefname{eq}{Eq.}{Eqs.}
\crefname{fig}{Fig.}{Figs.}
\crefname{section}{Sec.}{Secs.}
\crefname{appendix}{App.}{Apps.}
\crefname{proposition}{Prop.}{Props.}
\crefname{theorem}{Thm.}{Thms.}
\crefname{lemma}{Lem.}{Lems.}
\crefname{corollary}{Cor.}{Cors.}

\makeatletter 
\renewcommand\onecolumngrid{%
  \do@columngrid{one}{\@ne}%
  \def\set@footnotewidth{\onecolumngrid}%
  \def\footnoterule{\kern-6pt\hrule width 1.5in\kern6pt}%
}
\makeatother

\crefformat{table}{Table #2#1#3}
\crefformat{equation}{Eq.(#2#1#3)}
\crefformat{figure}{Figure #2#1#3}
\crefformat{section}{Sec.(#2#1#3)}
\crefformat{appendix}{Appendix #2#1#3}
\crefformat{proposition}{Proposition #2#1#3}
\crefformat{theorem}{Theorem #2#1#3}
\crefformat{lemma}{Lemma #2#1#3}
\crefformat{corollary}{Corollary #2#1#3}
\crefrangeformat{equation}{eqs.~(#3#1#4)~and~(#5#2#6)}

\usepackage{nameref}

\begin{document}
\title{Accidental Symmetry in the Tavis-Cummings Model via the Schwinger Boson Representation}

\author{Plato Deliyannis}
\email{plato.deliyannis@duke.edu}
\affiliation{Duke Quantum Center and Department of Physics, Duke University, Durham, NC 27708, USA}
\author{Iman Marvian}
\email{iman.marvian@duke.edu}
\affiliation{Duke Quantum Center and Department of Physics, Duke University, Durham, NC 27708, USA}\affiliation{Department of Electrical and Computer Engineering, Duke University, Durham, NC 27708, USA}

\begin{abstract}
The Jaynes-Cummings (JC) Hamiltonian is a paradigmatic model of light-matter interaction and, more generally, qubit--boson interactions, widely used across atomic, optical, and superconducting qubit platforms.
In the multi-qubit setting, where $n$ qubits are identically coupled to a single boson mode, this interaction is known as the Tavis-Cummings (TC) Hamiltonian.
The structure of the TC model is usually understood in terms of two standard symmetries: permutation invariance of the qubits and a U(1) symmetry associated with conservation of the total excitation number.
Here we identify an additional, independent ``accidental'' symmetry of the TC Hamiltonian and construct the corresponding conserved observable.
We show that, for $n\ge 3$ qubits, this symmetry imposes strong constraints on the realizable unitary transformations.
These constraints persist in the presence of the global $J_z$ Hamiltonian, but are removed by adding $J_z^2$, even though $J_z^2$ preserves both permutation invariance and the U(1) symmetry.
Finally, we explain the origin of this previously unnoticed symmetry using Schwinger's boson representation of angular momentum.
These restrictions have important implications for controllability of the TC system and for its applications to quantum computing, which are investigated further in a companion paper.
\end{abstract}

\maketitle

\section{Introduction} \label{sec:intro}
The Jaynes-Cummings (JC) model provides a simple yet powerful description of the interaction between a qubit, i.e., a two-level system, and a bosonic mode, such as a quantum harmonic oscillator \cite{Jaynes_Cummings_1963,JC_history}. 
Qubit-bosonic interactions in many physical platforms, including trapped ions \cite{Leibfried_2003_trappedion,Haffner_2008_trappedion}, superconducting qubits in circuit QED \cite{Blais_2004_cQED,Wallraff_2004_cQED}, and atoms in cavity QED \cite{Raimond_2001_cavity,Walther_2006_cavity,Varcoe_2000_cavity}, are well approximated by the JC model. 
Because of this ubiquity, the JC model is widely employed as a tool for implementing universal quantum computation \cite{Childs_Chuang_2000,Yuan_Lloyd_2007}, for example by the Cirac-Zoller \cite{Cirac_Zoller_1995} and Mølmer-Sørensen \cite{Molmer_Sorensen} schemes.

In this work, we consider the natural multi-qubit generalization in which multiple qubits are \textit{identically} coupled to a single bosonic mode via the JC interaction, which is commonly referred to as the Tavis-Cummings (TC) interaction \cite{tavis_1968_exact_solut,TC2_1969}. 
For a system of $n$ qubits, this interaction is described by the Hamiltonian
\begin{align}
 \begin{split}
    \HTC &\eq \frac{\gTC}{2}\sum_{i=1}^{n} \Big(\sigma_+^{(i)}{a} + \sigma_-^{(i)}{a}^{\dag}\Big)\\[4pt]
    &\eq \gTC\,\big({J}_+{a} + {J}_-{a}^{\dag}\big)\,,
 \label{eq:TCham}
 \end{split}
\end{align}
where ${\sigma}_{\pm}^{(i)}:=\sigma_x^{(i)}\pm i\sigma_y^{(i)}$ are Pauli raising/lowering operators acting on qubit $i$,
\begin{align}
    J_{\pm} \eq \frac{1}{2}\sum_{i=1}^{n}\sigma_{\pm}^{(i)} \eq J_x\pm i J_y\,, 
\end{align}
are the corresponding $n$-qubit total angular momentum operators, $a$ is the annihilation operator on the bosonic Hilbert space satisfying $[a, a^\dag]=\mathbb{I}$, and $\gTC$ is the coupling strength.\footnote{As usual, we often suppress tensor products in the notation, e.g. $J_+\simeq J_+\otimes\1_{\text{osc}}$, and $J_+a\simeq J_+\otimes a$.}
Sometimes, we also consider the phase-shifted version of this Hamiltonian, namely
\footnote{In terms of the canonical position and momentum operators $x = \frac{1}{\sqrt{2}}\big(a+a^{\dag}\big)$ and $p = \frac{i}{\sqrt{2}}\big(a^{\dag}-a\big)$, we have
$\HTC(\phi):=\gTC\sqrt{2}\big(\cos(\phi)(J_xx-J_yp) + \sin(\phi)(J_yx+J_xp)\big).$}
\begin{align}
 \begin{split}
    \HTC(\phi) &\eq e^{i\phi J_z}\HTC e^{-i\phi J_z} \\[4pt]
    &\eq \gTC\,\big(e^{i\phi}{J}_+{a} + e^{-i\phi} {J}_-{a}^{\dag}\big)\,,
 \end{split}
\end{align}
for $\phi\in[0,2\pi)$. 
Then, the full Hamiltonian of the qubits and bosonic mode can be written as
\begin{align}
    H(t) \eq \omega_{\text{osc}}(t) a^\dag a + \omega_z(t) J_z + f(t) \HTC\,,  
\label{eq:general_Ham}
\end{align}
where $\omega_{\text{osc}}(t)$ and $\omega_z(t)$ are, respectively, the bosonic mode frequency and the qubit transition frequency, and $f(t)$ is an arbitrary real function that determines the strength of the interaction.
\begin{figure}
    \centering
    \includegraphics[width=0.96\linewidth]{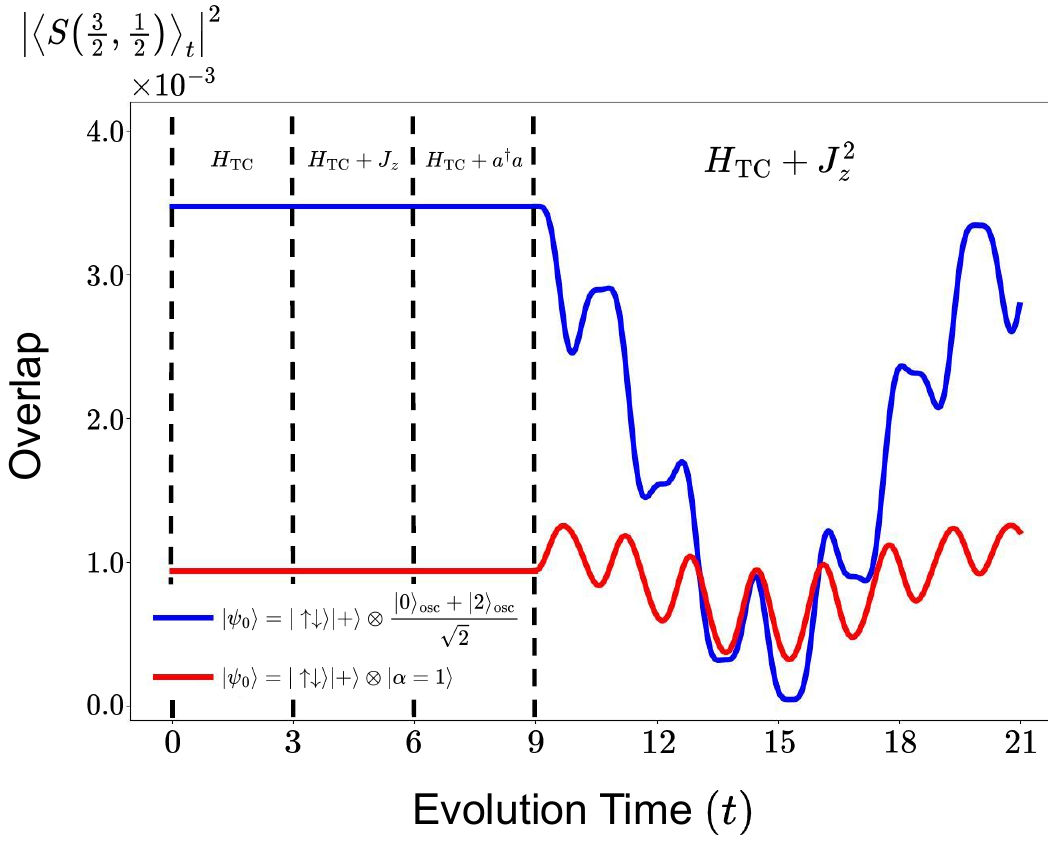}
    \caption{\textbf{Conservation law associated with the accidental symmetry.} 
    Example of a $3$-qubit system illustrating the conservation law, which holds for time evolution under Hamiltonians $\HTC$, $J_z$, and $a^{\dag}a$, but is violated by  $\HTC+J_z^2$. 
    The vertical axis shows the absolute value of the expectation 
    $|\ipo{\psi(t)}{S(\sfrac{3}{2},\sfrac{1}{2})}{\psi(t)}|$, where operator $S(\sfrac{3}{2},\sfrac{1}{2})$ is defined in \cref{sec:conserved_observable}.   
    Qubits are initialized in the product state $\ket{\uparrow} \ket{\downarrow} \ket{+}$, with $\ket{+} = (\ket{\uparrow} + \ket{\downarrow})/\sqrt{2}$, and $\ket{\uparrow},\ket{\downarrow}$ eigenstates of $\sigma_z$ with eigenvalues $+1$ and $-1$, respectively.
    The bosonic mode is initialized either in the coherent state $\ket{\alpha=1}= e^{-1/2}\sum_{n=0}^{\infty}  \ket{n}_{\text{osc}}/\sqrt{n!}$ (red curve) or in $(\ket{0}_{\text{osc}} + \ket{2}_{\text{osc}})/\sqrt{2}$ (blue curve), where $\{\ket{n}_{\text{osc}}\}$ is the occupation (also known as Fock) basis.}
    \label{fig:cons_law}
\end{figure}

This model has been studied extensively in various contexts, from atomic physics to quantum control \cite{Bashir_Abdalla_1995,Bogoliubov_etal_1996,Rybin_etal_1998,Tessier_etal_2003,Vadeiko_etal_2003,Genes_etal_2003,Fink_etal_2009,Agarwal_etal_2012,Zou_etal_2013,Keyl_2014_control}.
Importantly, the permutational symmetry of $\HTC$ allows it to be implemented using collective control fields that act identically on all qubits, which are often easier to engineer than individually addressing each qubit. 
For instance, the TC Hamiltonian can be realized through the coupling of a linear trapped-ion chain to its center-of-mass vibrational mode \cite{Retzker_2007}, or in superconducting qubits \cite{Fink_etal_2009}.

Because the TC Hamiltonian inherently couples multiple qubits to a single bosonic mode, it provides a natural platform for engineering many-body interactions, a broadly useful resource in quantum computing.
Indeed, protocols for preparing highly entangled states, such as Dicke and GHZ states, have been developed by exploiting this interaction \cite{Retzker_2007,Mu_etal_2020,Zhang_etal_2024}.

\subsection{Permutational and U(1) symmetries}
Since the original work by Tavis and Cummings \cite{tavis_1968_exact_solut}, studying the TC Hamiltonian begins by considering its two ``standard'' symmetries: (1) invariance with respect to permuting the qubits, and (2) a U(1) symmetry corresponding to a conserved charge given by the total number of excitations in the bosonic mode and qubits, namely the \textit{charge operator}
\begin{align}
    Q \eq a^{\dag}a + J_z + \frac{n}{2}\,.
 \label{eq:Qop}
\end{align}
(The constant offset $\tfrac{n}{2}$ is included, so the eigenvalues of $Q$ are the non-negative integers, $q=0,1,2,\dotsc\infty$.)
To see that $\HTC$ conserves this charge, observe that the interaction term $J_+a$ raises the $z$-component of angular momentum while lowering the bosonic excitation number by one, while $J_- a^\dag$ produces the opposite effect.
In addition, the harmonic oscillator and qubit Hamiltonians $a^\dag a$ and $J_z$ also conserve this charge.
Therefore,
\begin{align}
   \big[H(t), Q\big]=0\,.
\end{align}

\subsection{An accidental symmetry}
In this work, we identify and characterize an additional “accidental” symmetry of the TC Hamiltonian, independent of its two ``standard'' symmetries.
\cref{eq:trans1} in \cref{sec:accidental_sym} provides a simple description of this symmetry, which, to the best of our knowledge, has not been noticed previously. 
Given the widespread relevance of TC interactions in quantum computing and quantum control, understanding this symmetry and its consequences for system dynamics is essential.

We analyze how this symmetry constrains the evolution of systems governed by Hamiltonians $H(t)$ in Eq.~(\ref{eq:general_Ham}) with arbitrary, time-dependent coefficients. 
We find that for systems with $n \ge 3$ qubits, it imposes nontrivial restrictions and conservation laws, revealing structure beyond the familiar permutational and U(1) symmetries.

To illustrate, \cref{fig:cons_law} shows a three-qubit system evolving under $\HTC$, $\HTC + J_z$, and $\HTC + a^\dag a$. 
In each case, the magnitude of the expectation value of a specific operator (defined in \cref{sec:conserved_observable}) remains conserved. 
By contrast, when the Hamiltonian $\HTC + J_z^2$ is applied, this quantity begins to evolve, despite the fact that permutational and U(1) symmetries are still preserved.
This demonstrates that the accidental symmetry and the corresponding conservation law are not an inevitable consequence of these standard symmetries.

Remarkably, we find that this symmetry admits a simple explanation through Schwinger's oscillator model of angular momentum \cite{Schwinger_model_1952, Sakurai_Napolitano_2020}, as discussed in \cref{sec:schwinger}.
We also explore the implications of this symmetry for the controllability of qubit-bosonic systems coupled via TC interactions in \cref{sec:controllability}.
Further details and discussions are presented in our first companion paper \cite{theory_paper}, while applications to quantum computing -- including explicit constructions of two-qubit gates -- are discussed in the second companion paper \cite{circuit_paper}.

\section{Preliminaries: Matrix elements of $\HTC$}
We begin by reviewing the matrix elements of the $\HTC$ in a convenient basis (see, e.g., \cite{tavis_1968_exact_solut, 
Keyl_2014_control}).
Recall that although $\HTC$ breaks the full rotational SU(2) symmetry of the qubits, and in particular it does not commute with $J_x$, its permutational symmetry ensures that the total angular momentum squared operator $J^2=J_x^2+J_y^2+J_z^2$ is still conserved, i.e.,
\begin{align}
    \big[\HTC, J^2\big]=0\,.
\end{align}
Clearly, the other terms in Hamiltonian $H(t)$ in \cref{eq:general_Ham} also satisfy this property.
This is a special case of a more general fact: any permutation-invariant (PI) operator acting on $(\mathbb{C}^2)^{\otimes n}$ commutes with any SU(2)-invariant operator on $(\mathbb{C}^2)^{\otimes n}$, such as $J^2$, a result that can be understood as a consequence of Schur-Weyl duality \cite{schur_1927,weyl_1939,goodman2010symmetry}.

Hence, to study $\HTC$, one often uses the total angular momentum orthonormal eigenbasis of $(\mathbb{C}^2)^{\otimes n}$, namely  $\ket{j,m,\alpha}$, satisfying
\begin{align} 
 \begin{split} 
    J^2 \,\ket{j,m,\al} &= j(j+1)\,\ket{j,m,\al}\\[4pt] 
    J_z \,\ket{j,m,\al} &= m\,\ket{j,m,\al}\,,
 \end{split}
\end{align} 
where $j$ is integers $j_{\text{min}}=0,1,\dotsc, n/2$ for even $n$, and half integers $j_{\text{min}}=1/2, 3/2, \dotsc, n/2$ for odd $n$; $m=-j,-j+1,\dotsc, j$ is the $z$-component of total angular momentum; and $\alpha$ labels the multiplicity, which corresponds to different copies of SU(2) irreducible representations.
In particular, the number of different values it can take,
\begin{align}
    M(n,j) \,:=\, \left(\begin{matrix}n\\ \frac{n}{2}-j\end{matrix}\right)\frac{2j+1}{\frac{n}{2}+j+1}\,,
 \label{eq:mnj}
\end{align}
is independent of $m$, and only depends on $n$ and $j$ \cite{MLH_2024,Bartlett_etal_2007}.
For a given $j$ and $\al$, one can obtain all $2j+1$ basis states for an irreducible representation of SU(2), using the ladder operators $J_{\pm}$, which act as
%
\begin{align} \label{eq:J+}
    J_+\ket{j,m,\al}&=\sqrt{(j+m+1)(j-m)}\,\ket{j,m+1,\al}\, \\[4pt]\label{eq:J-}
    J_-\ket{j,m,\al}&=\sqrt{(j-m+1)(j+m)}\,\ket{j,m-1,\al}\,.
\end{align}
%
For example, one can choose
\begin{align}\label{eq:ex}
    \ket{j,j,\al_0} := \ket{\Psi^-}^{\otimes (n/2-j)}\otimes\ket{\uparrow}^{\otimes2j}\,,
\end{align}
where
\begin{align}
    \ket{\Psi^-}=\frac{1}{\sqrt{2}}\Big(\ket{\uparrow}\ket{\downarrow}-\ket{\downarrow}\ket{\uparrow}\Big)
\end{align}
denotes a 2-qubit singlet and $\ket{\uparrow}$ and $\ket{\downarrow}$ are the $+1$ and $-1$ eigenstates of $\sigma_z$, respectively. 
Then, all basis elements $\ket{j,m,\al_0}$ with the same values of $j$ and $\alpha_0$, and lower  $m\le j$ can be obtained via the lowering operator $J_-$, as 
\begin{align}\label{lowering}
    \ket{j,m,\al_0} = c_{j,m} (J_-)^{j-m}\ket{j,j,\al_0}\,,
\end{align}
where $c_{j,m}$ is a normalization coefficient, namely
\begin{align}
    c_{j,m}^{-1} &:= \prod_{\widetilde{m}=m+1}^{j}\sqrt{(j+\widetilde{m})(j-\widetilde{m}+1)}\,.
\end{align}
An example with $n=3$ qubits is given in \cref{sec:3Q_example}.
For the bosonic mode, the occupation basis (also known as the Fock basis) is defined by
\begin{align}
    a^{\dag}a\ket{k}_{\text{osc}} = k\ket{k}_{\text{osc}}\quad:\quad k=0,\,\dotsc,\infty\,;
\end{align}
which satisfies, via $[a,a^{\dag}]=\1$,
\begin{align}
    a^\dag\ket{k}&=\sqrt{k+1} \ket{k+1}\,.
 \label{eq:adag}
\end{align}
Then, a convenient basis for the total Hilbert space of $n$ qubits and a bosonic mode is
\begin{align}\label{eq:basis}
    \ket{j,m,\al}\otimes\ket{k}_{\text{osc}}\,.
\end{align}
The explicit matrix of $\HTC$ in this basis is symmetric and tri-diagonal with zeros along the diagonal.
In particular, using \Cref{eq:J+,eq:J-,eq:adag}, the \textit{non-zero} matrix elements of $\HTC$ are 
\begin{align}
 \begin{split}
    &\big(\bra{j, m,\alpha}\otimes \bra{k}_{\text{osc}}\big)\,\HTC\,\big(\ket{j,m-1,\alpha}\otimes\ket{k+1}_{\text{osc}}\big) \\[2pt]
    &=\sqrt{k+1}\times \sqrt{(j+m)(j-m+1)}\,,
 \label{eq:HTC_elements}
 \end{split}
\end{align}
and its equal Hermitian conjugate term.
It is significant to note that the matrix elements depend on $j$, $m$, and $k$, but not the multiplicity index $\al$ or number of qubits $n$.

Note also that the charge operator $Q := a^{\dag}a + J_z + n/2$ is diagonal in this basis, with its eigenvalues given by
\begin{align}
    q \eq k + m + \frac{n}{2}\,.
\end{align}
In summary, the permutational symmetry of $\HTC$ is reflected in the fact that for any pair of basis elements $\ket{j_1,m_1,\alpha_1} \otimes \ket{k_1}$ and $\ket{j_2,m_2,\alpha_2} \otimes \ket{k_2}$, the corresponding matrix element vanishes unless $j_1=j_2$ and $\alpha_1=\alpha_2$, and moreover it is independent of $\alpha_1$.
Similarly, the U(1) symmetry is respected because the matrix elements vanish unless $m_1+k_1=m_2+k_2$.
Any Hamiltonian with these two properties satisfies both permutational and U(1) symmetry.
In particular, the bosonic Hamiltonian $a^\dag a$ and qubit Hamiltonian $J_z$ both satisfy these criteria. 
Similarly, the Hamiltonian $J_z^2$ also respects this symmetry.

\section{The Accidental Symmetry of $\HTC$}
\label{sec:accidental_sym}
Next, we show that $\HTC$ satisfies an additional symmetry that is independent of the permutational  and U(1) symmetries; that is, it is not shared by all Hamiltonians that respect those two symmetries, such as $J_z^2$. 
The accidental symmetry is a consequence of the following simple but consequential observation: the matrix elements of $\HTC$ remain invariant under the transformation
\begin{align}\label{eq:trans1}
 \begin{split}
    \begin{pmatrix}
        j \\
        m \\
        k
    \end{pmatrix}\longrightarrow 
    \begin{pmatrix}
        j' \\
        m' \\
        k'
    \end{pmatrix}
    =
    \begin{pmatrix}
        \sfrac{1}{2} & \sfrac{1}{2} & \sfrac{1}{2} \\
        \sfrac{1}{2} & \sfrac{1}{2} & -\sfrac{1}{2} \\
        1 & -1 & 0
    \end{pmatrix}
    \begin{pmatrix}
        j \\
        m \\
        k
    \end{pmatrix}\,,
 \end{split}
\end{align}
as long as $j'$, $m'$, and $k'$ remain in their allowed ranges, i.e. $j'\leq n/2$ and $|m'|\leq j'$.
To ensure that $j'$ and $m'$ are integers (half-integers) for even (odd) $n$, we should only consider even (odd) $k$ in the above transformation.
This map is therefore valid only for certain values of $j,m,k$, which we discuss further below, in \cref{sec:invariant}.

\vspace{3mm}
More explicitly, recall that the only non-zero matrix elements of $\HTC$ are given by \cref{eq:HTC_elements}, and its Hermitian conjugate. 
Under the transformation in \cref{eq:trans1}, the right-hand side of \cref{eq:HTC_elements} is mapped to
\begin{align}
 \begin{split}
    \sqrt{k'+1}&\times \sqrt{(j'+m')(j'-m'+1)} \\
    &= \sqrt{j-m+1}\times\sqrt{(j+m)(k+1)}\,,
 \end{split}
\end{align}
which means it remains unchanged.

A few remarks are in order. 
First, it is worth noting that the map in \cref{eq:trans1} mixes the qubit degrees of freedom $j$ and $m$ with those of the bosonic mode, $k$.
Moreover, it is an involution, meaning that applying it twice yields the identity, which suggests that it \textit{swaps} two sets of states. 
As we show in \cref{sec:schwinger}, this involution acquires a natural explanation when the accidental symmetry is interpreted using the Schwinger representation of angular momentum.
(We also note that this $3\times 3$ matrix is not a unitary transformation.)

Additionally, note that all states with the same angular momentum $j$ and the same charge $q=k+m+n/2$, are mapped to states with a particular angular momentum $j'$ and charge $q'=k'+m'+n/2$, given by
\begin{align}\label{eq:trans2}
 \begin{split}
    \begin{pmatrix}
        j \\
        q
    \end{pmatrix}\longrightarrow 
    \begin{pmatrix}
        j' \\
        q' 
    \end{pmatrix}
    = \frac{1}{2}
    \begin{pmatrix}
        1 & 1 \\
        3 & -1 
    \end{pmatrix}
    \begin{pmatrix}
        j \\
        q
    \end{pmatrix}
    +\frac{n}{4} 
    \begin{pmatrix}
        -1 \\
        3
    \end{pmatrix} \,.
 \end{split}
\end{align}
This expression follows from \cref{eq:trans1}, by projecting to states for which $k+m$ is constant -- equivalently, multiplying by the matrix $\big(\begin{smallmatrix}
    1&0&0\\0&1&1
\end{smallmatrix}\big)$ -- and using the definition $q:=k+m+n/2$.
Additionally, rearranging \cref{eq:trans2}, such that $q$ and $q'$ are written in terms of $j$ and $j'$ yields
\begin{align} \label{eq:qq'_matrix}
 \begin{split}
    \begin{pmatrix}
        q \\
        q'
    \end{pmatrix}
    =
    \begin{pmatrix}
        -1 & 2 \\
        2 & -1
    \end{pmatrix}
    \begin{pmatrix}
        j \\
        j'
    \end{pmatrix}
    +\frac{n}{2} 
    \begin{pmatrix}
        1 \\
        1
    \end{pmatrix}\,.
 \end{split}
\end{align}
In other words, for each pair  of angular momenta $j$ and $j'$, the map in \cref{eq:trans1} relates states with corresponding charges $q$ and $q'$, as given in \cref{eq:qq'_matrix}, which satisfy
\begin{align}
    q'-q=3\times (j-j')\,.
\end{align}
Furthermore, the transformation in \cref{eq:trans1} preserves the sum
\begin{align}
    r \,:=\, j+m\,,
\end{align}
which also means that the ordering of basis elements $\ket{j,m,\al}\otimes\ket{k}_{\text{osc}}$ with respect to $m$ is preserved.
For example, a state with minimal $z$-component of angular momentum $\ket{j,m=-j}$ is mapped to another state with minimal $z$-component of angular momentum, $\ket{j',m'=-j'}$.

Before providing a full formal description of this accidental symmetry, it is useful to consider a different labeling for the basis we defined above.

\subsection{Invariant subspaces of PI, U(1)-invariant operators} \label{sec:invariant}
Recall the orthonormal basis $\ket{j,m,\alpha}$ for $(\mathbb{C}^2)^{\otimes n}$, and the orthonormal basis
\begin{align}\nonumber
    \ket{j,m,\alpha}\otimes \ket{k}_{\text{osc}} \,,
\end{align}
of the total Hilbert space $\H_{\text{qubits}}\otimes\H_{\text{osc}}$. 
The subspace with charge $q$ and angular momentum $j$ corresponds to all states $\ket{j,m,\alpha}\otimes\ket{k}_{\text{osc}}$, satisfying the constraints $q=m+k+n/2$, where $k\ge 0$ and $-j\le m\le j$.
Putting these constraints together, we find that $m$ takes values
\begin{align}
    m= -j,\,\dots,\,\min\left(j,\,q-\frac{n}{2}\right)\,.
\end{align}
Note that the minimum value of $m$ is $-j$ and the minimum value of $k$ is zero, which means for each $j$, the allowed charge values are $q\ge n/2-j$. 
Equivalently, for each $q$, the allowed $j$ values are $j\ge n/2-q$. 
Of course, the minimum value of $j$ cannot be less than $j_{\text{min}}$, which implies that for a given value of $q$, the allowed values of $j$ are
\begin{align}
    j \,\geq\, \max\left(\jmin,\,\frac{n}{2}-q\right)\,.
\end{align}
Roughly speaking, this is because achieving lower values of $q < n/2$ requires more qubits to be in the state $\ket{\downarrow}$, meaning they tend to be aligned, which corresponds to a higher value of $j$.

In the following, instead of eigenvalues of $J_z$ and $a^\dag a$, it is useful to label elements of basis in \cref{eq:basis} with $q$ the eigenvalue of the charge operator $Q$, and the integer $r=j+m$ as
\begin{align} \label{eq:r_qja}
 \begin{split}
    \big|r^{(q,j,\alpha)}\big\rangle
    &\,:=\, \ket{j,\,r-j,\,\alpha}\otimes \left|q+j-\frac{n}{2}-r\right\rangle_{\rm osc} \\[6pt]
    &\hspace{60pt}:\,\,r=0,\dotsc, d_{n}(q,j)-1\,,
 \end{split}
\end{align}
where
\begin{align}
    d_{n}(q,j):=\min\left\{2j+1,\,q+j-\frac{n}{2}+1\right\}\,.
\end{align}
For fixed values of $q$, $j$, and $\alpha$, this defines a subspace:
\begin{align} \label{eq:sector_def}
 \begin{split}
    \mathcal{H}_{q,j} \eq\text{span}\Big\{&\big|r^{(q,j,\al)}\big\rangle =\ket{j,m,\alpha}\otimes \ket{q-m-\tfrac{n}{2}}_{\text{osc}} \\[2pt] &\,\,:\,\, m= -j,\dotsc, \min\Big(j, q-\frac{n}{2}\Big)\Big\}\,,
 \end{split}
\end{align}
henceforth referred to as a \textit{sector} (see further, \cref{sec:sector_decomp}), with dimension
\begin{align}
 \begin{split}
    \text{dim}(\H\qj) &\eq d_n(q,j)\\[6pt]
    &\eq\begin{cases}
    2j+1 & :\,\,q\geq n/2+j\\[4pt]
    q+1+j-\frac{n}{2} & :\,\,q< n/2+j\,.
  \end{cases}
 \end{split}
\label{eq:recall_dimHqj}
\end{align}
We refer to these two cases as \emph{filled} and \emph{unfilled} sectors, respectively, as in the former case, $\H\qj$ contains the entire $(2j+1)$-dim subspace corresponding to an irrep of SU(2) with angular momentum $j$ (see \cref{fig:6_qubit}).

Any operator respecting the permutational and U(1) symmetries commutes with $J^2$ and $Q$, and is block-diagonal with respect to the subspaces defined in \cref{eq:sector_def}.
Then, for any such operator $A$, we can denote its matrix elements with respect to \cref{eq:r_qja} as
\begin{align} \label{eq:r_matrix_els}
    \big[A_{q,j}\big]_{r_1,r_2} \,:=\, \big\langle r_1^{(q,j,\alpha)}\big|\,A\,\big|r_2^{(q,j,\alpha)}\big\rangle\,,
\end{align}
where the RHS is independent of $\al$ because of permutational symmetry.

It is worth noting that the integer $r=j+m$, can be thought of as the eigenvalue of observable
\begin{align}
    \sqrt{J^2+\sfrac{1}{4}}-\sfrac{1}{2} + J_z \,.
\end{align}

\subsection{Formal statement of accidental symmetry}
The following proposition gives a formal statement of this accidental symmetry, excluding the case of one-dimensional invariant sectors, which will be discussed separately below.
\begin{proposition} \label{prop:main_prop}
For each pair of angular momenta $j$ and $j'$ satisfying
\begin{align}
    0 \,<\,j'\,<\,j\,\leq\,\frac{n}{2}\,,
\end{align}
and for each choice of corresponding multiplicity labels $\alpha$ and $\alpha'$, there exists a \textbf{unique} pair of invariant subspaces of
the Tavis-Cummings interaction $\HTC(\phi)$, for all $\phi\in[0,2\pi)$, each of dimension $2j'+1$ with the following properties:

The first subspace is an unfilled sector, contained in the subspace with angular momentum $j$ and charge
\begin{align} \label{eq:q}
    q \eq \frac{n}{2}+2j'-j\,,
\end{align}
and spanned by the orthonormal basis
\begin{align}\label{eq:input-subspace}
 \begin{split}
    \big|r^{(q,j,\alpha)}\big\rangle
    &\,=\, \ket{j,\,r-j,\,\alpha}\otimes \ket{2j'-r}_{\rm osc}\\[4pt]
    &\hspace{56pt}:\quad r=0,\ldots,2j'\,.
 \end{split}
\end{align}

The second subspace is a filled sector, contained in the subspace with angular momentum $j'$ and charge
\begin{align}
 \label{eq:q'}
    q' \eq \frac{n}{2}+2j-j'\,,
\end{align}
and spanned by the orthonormal basis
\begin{align} \label{eq:output-subspace}
 \begin{split}
    \big|r^{(q',j',\alpha')}\big\rangle
    &\,=\, \ket{j',\,r-j',\,\alpha'}\otimes \ket{2j-r}_{\rm osc}\\[4pt]
    &\hspace{56pt}:\quad r=0,\ldots,2j'\,.
 \end{split}
\end{align}
Moreover, the matrix elements of $\HTC(\phi)$ relative to these two bases are identical; specifically, the only nonzero matrix elements are:
\begin{align} \label{eq:HTCphi_matrix}
 \begin{split}
    &\big[\HTC(\phi)_{q',j'}\big]_{r,\,r-1} \eq \big[\HTC(\phi)_{q,j}\big]_{r,\,r-1}\\[4pt]
    &\hspace{34pt} \eq e^{i\phi}\sqrt{r(2j-r+1)(2j'-r+1)}\,,
 \end{split}
\end{align}
and their Hermitian conjugates, where $r=1,\ldots,2j'$.
Thus the restrictions of $\HTC(\phi)$ to these two subspaces are unitarily equivalent.
\end{proposition}
In the general $n$-qubit case, total angular momentum $j$ takes $\lceil\frac{n}{2}\rceil$ values, excluding $j=0$.
Therefore, it follows that there are
\begin{align}
    {{\lceil\frac{n}{2}\rceil}\choose{2}} \eq \frac{1}{2}\times \left\lceil\frac{n}{2}\right\rceil\times \left(\left\lceil\frac{n}{2}\right\rceil-1\right)\,
\end{align}
pairs of sectors related by the accidental symmetry -- i.e. the number of unique pairs of integers (or half-integers) $(j,j')$ satisfying $0<j'<j\leq\frac{n}{2}$, for even (or odd) $n$, respectively.
We return to this key fact when discussing the controllability of the TC model in \cref{sec:controllability}.

For example, in \cref{sec:3Q_example}, we discuss the case of $n=3$ qubits, for which there is exactly one such pairwise equivalence: $j=3/2$ and $j'=1/2$.
\Cref{fig:6_qubit} illustrates the case of $n=6$ qubits, for which there are three such pairs.
\begin{figure}[htp]
    \centering
    \includegraphics[width=0.96\linewidth]{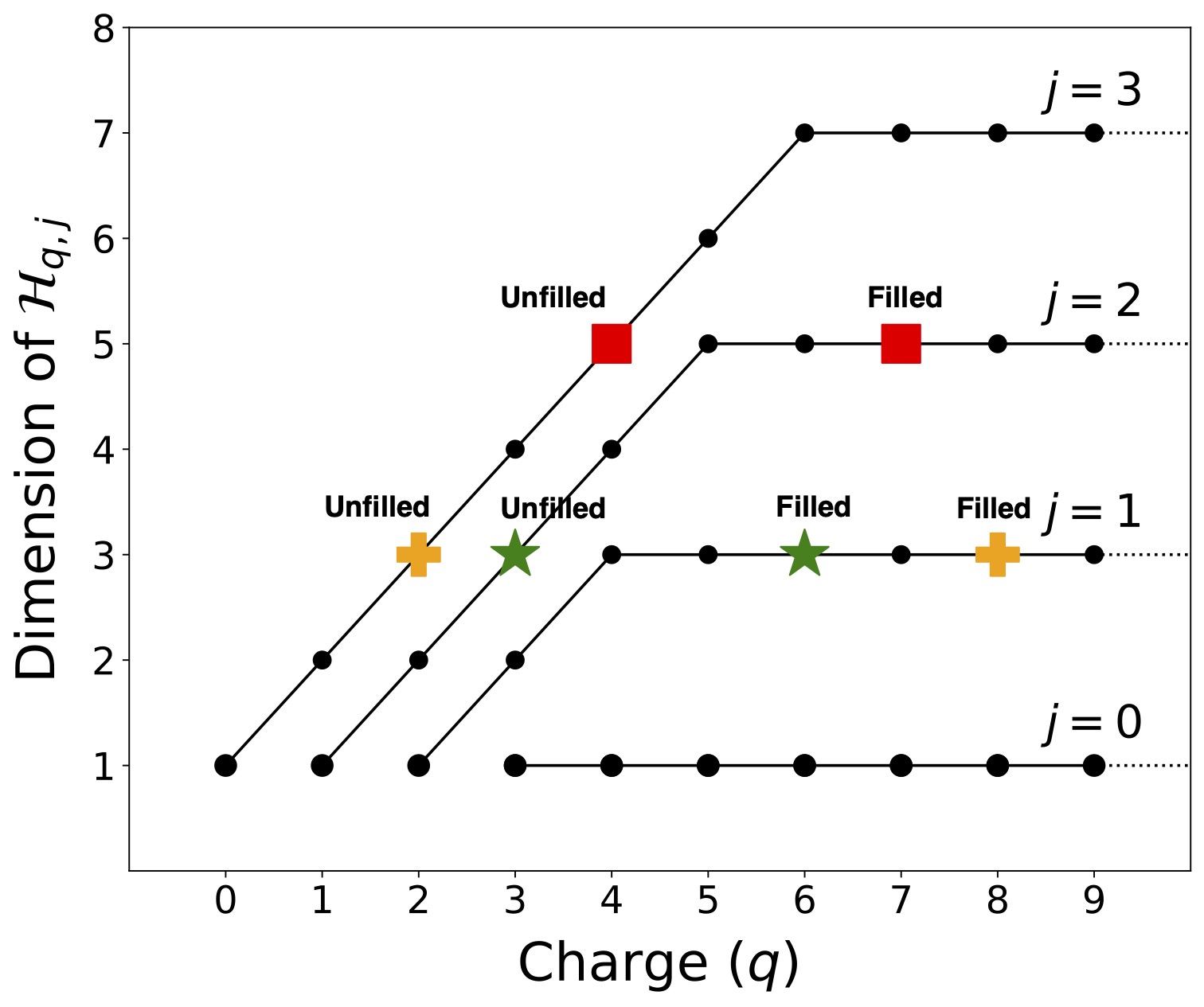}
    \caption{\textbf{Accidental symmetry of TC Hamiltonian:} (6-qubit example). For each pair of angular momenta $j>j'$, there is a single pair of sectors $\H\qj$ and $\H_{q',j'}$ related by the accidental symmetry.
    Diagrammatically, these sectors lie along the horizontal line which represents dimension $2j'+1$: the maximum dimension allowed for a sector with angular momentum $j'$.
    Observe that for each $j>j'>0$, there is a single angular momentum $j$ sector $\H\qj$ with this dimension.
    Then, the difference in charges $q$ and $q'$ is $q'-q=3(j-j')$, i.e. thrice the difference in angular momenta. 
    For example, the red square indicates the pair for $j=3$ and $j'=2$, the yellow cross $j=3$ and $j'=1$, and the green star $j=2$ and $j'=1$.}
    \label{fig:6_qubit}
\end{figure}

\subsection{Uniqueness of paired accidental symmetry sectors} \label{sec:uniquess}
It is worth emphasizing that within the angular momentum $j$ subspace, there are infinitely many filled sectors $\H\qj$ with dimension $2j+1$.
However, as we show below, the components of $\HTC(\phi)$ in different filled angular momentum $j$ sectors cannot be unitarily equivalent.
Furthermore, we show that the components of $\HTC(\phi)$ in different \textit{unfilled} sectors with the same dimension cannot be unitarily equivalent either.
Therefore, the only possibly for components of $\HTC(\phi)$ in different sectors being unitarily equivalent is precisely pairs of one filled and one unfilled sector; and where any given sector appears in at most one such pair.

To see this, note that if the components of $\HTC(\phi)$ in two different sectors were unitarily equivalent, their second moments, defined as
\begin{align} \label{eq:second_moment_def}
    \Tr\Big\{\big[\HTC^2(\phi)_{q,j}\big]\Big\} \eq \sum_{r=0}^{d_{n}(q,j)} \big[\HTC^2(\phi)_{q,j}\big]_{r,r}\,,
\end{align}
would also be equal.
However, using the matrix elements in \cref{eq:HTC_elements}, we show in Appendix~(\ref{sec:appendix}) that for filled sectors with fixed $j$, the second moments of $\HTC(\phi)$ are monotonically increasing with $q$.
Intuitively, this is due to the fact that sectors with higher number of excitations have higher energies.
On the other hand, for any fixed $j$, the dimension of the unfilled sectors $\mathcal{H}_{q,j}$ linearly grows with charge $q$, so there is only one unfilled sector with a given angular momentum $j$ and dimension $2j'+1$.
We also show that for any pair of two unfilled sectors with the same dimension, $\HTC(\phi)$ has distinct second moments.
Thus, the only possibility is exactly one filled and one unfilled sector.
See Appendix~(\ref{sec:appendix}) for further details.
\begin{figure}[htp]
    \centering
    \includegraphics[width=0.96\linewidth]{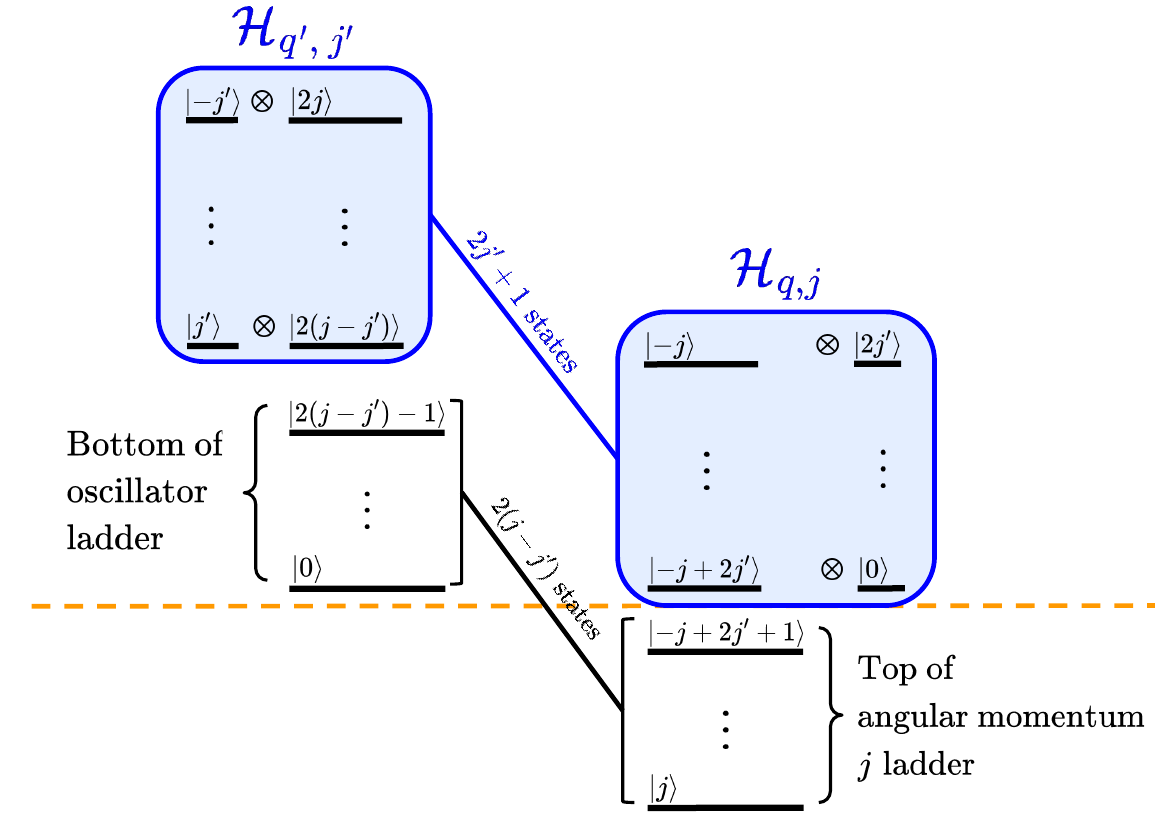}
    \caption{Illustration of a pair of sectors related by the accidental symmetry, as described in \cref{prop:main_prop}.
    Sector $\H\qj$ is unfilled, while sector $\H_{q',j'}$ is filled, meaning it carries a full irrep of SU(2).
    The matrix elements of $\HTC$ in these two sectors, given in \cref{eq:HTCphi_matrix}, are identical.}
    \label{fig:pairwise_equiv}
\end{figure}

\subsection{One-dimensional invariant subspaces}
\label{sec:1D}
In the above discussion, we ignored the one-dimensional (1D) subspaces that are joint eigenvectors of $J_z$, $\HTC$ and hence also joint eigenvectors of $\HTC(\phi)$.

Any operator that respects both permutational symmetry and the U(1) symmetry, associated with the charge $Q$, is block diagonal in the decomposition into sectors $\mathcal{H}\qj$.
Therefore, every 1D sector $\mathcal{H}\qj$ is necessarily an eigenvector of such Hamiltonians.
These sectors are characterized by the condition
\begin{align}
    d_n({q,j})=\min\left\{2j+1,\,q+j-\frac{n}{2}+1\right\}=1\,,
\end{align}
which is satisfied in two cases: (i) $j=m=0$, but arbitrary $q=k+n/2$, corresponding to SU(2) singlets and arbitrary eigenstates of the bosonic mode, or (ii)
\begin{align}
    q \eq \frac{n}{2}-j\,,
\end{align}
in which case $k=0$ and $m=-j$, corresponding to states at the bottom of the SU(2) ladder with the bosonic mode in the vacuum state (see \cref{fig:2q_inv_subs}).
In particular, the former class of states are eigenstates because $j=0$ states, or singlets, are invariant under SU(2).
The latter class of states are eigenstates because the bosonic mode cannot be lowered further, since the minimum eigenvalue of $a^\dag a$ is zero, and the qubits cannot be lowered further within the same permutation-symmetry sector.
Indeed, although lower values of $m$ may exist down to $m=-n/2$, they belong to different angular momentum sectors and therefore cannot be reached without breaking the permutational symmetry $\mathbb{S}_n$.
\begin{figure}[htp]
    \centering
    \includegraphics[width=0.96\linewidth]{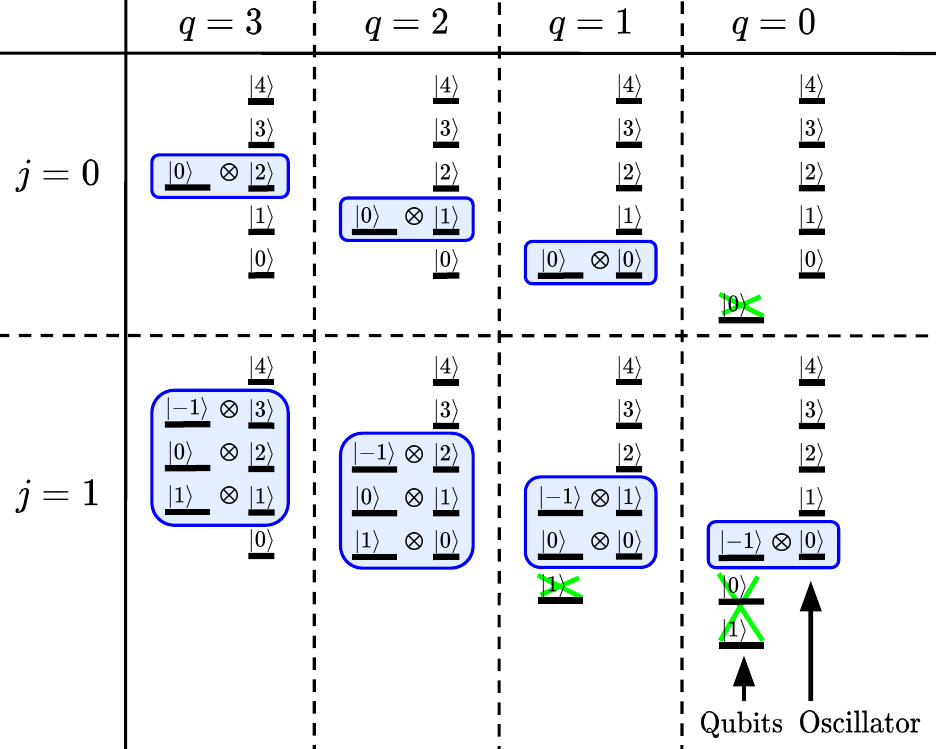}
    \caption{\textbf{2-qubit example:}
    Blue boxes indicate invariant subspaces of any PI, U(1)-invariant operator acting on $n=2$ qubits coupled to a bosonic mode.
    The states are in the format $\ket{m}\otimes\ket{k}$, where $m$ and $k$ are eigenvalues of $J_z$ and $a^{\dag}a$, respectively.
    Note that there are two types of 1D invariant sectors: states in the form $\ket{j=0,m=0}\otimes\ket{k}$ for all $k\ge 0$, and the state $\ket{j=1,m=-1}\otimes\ket{k=0}$.
    All of these 1D sectors have eigenvalue $0$ of $\HTC$.
    The accidental symmetry pairs $\ket{j=1,m=-1}\otimes\ket{k=0}$ and
    $\ket{j=0,m=0}\ket{k=2}$ and implies they should have the same eigenvalue, a constraint that is not necessarily satisfied by arbitrary PI, U(1)-invariant Hamiltonians.}
 \label{fig:2q_inv_subs}
\end{figure}

It happens that each of these 1D sectors is an eigenstate of $\HTC$ with eigenvalue zero.
However, although every PI, U(1)-invariant operator necessarily features these 1D eigen-subspaces, unlike the case of $\HTC$, the eigenvalues of a general PI, U(1)-invariant operator $A$ could be arbitrary and independent of each other.
However, if $A$ also respects the accidental symmetry, then when $n$ is even, certain pairs of eigenvalues must be identical. 
In particular, the map in \cref{eq:trans1} connects  one eigenvector with $j=0$ to one eigenvector with $m=-j$ and $k=0$, namely
\begin{align}
    \ket{j,\,m=-j}\otimes\ket{0}_{\text{osc}}\,\longleftrightarrow\,\ket{j=0,\,m=0}\otimes\ket{2j}_{\text{osc}}\, .
\end{align}
This means that they have the same eigenvalue for any operator that respects the accidental symmetry.

\subsection{Further Remarks}
Under the transformation in \cref{eq:trans1}, $J_z$ does not remain invariant; rather its nonzero matrix elements satisfy:
\begin{align} \label{eq:Jz'_matrix}
 \begin{split}
    \big[(J_z)_{q',j'}\big]_{r_1,r_2} \eq \big[(J_z)_{q,j}\big]_{r_1,r_2} + (j-j')\delta_{r_1,r_2}\,.
 \end{split}
\end{align}
In other words, for a set of states with fixed $q$ and $j$, \cref{eq:trans1} shifts the $z$-component of angular momentum by a constant $(j-j')$, independent of $m$, $k$, and $\al$.
Similarly, for a set of states with fixed $q$ and $j$, \cref{eq:trans1} shifts the bosonic excitation number by a constant $2(j-j')$, independent of $m$, $k$, and $\al$, i.e.
\begin{align}
 \begin{split}
    &\big[(a^{\dag}a)_{q',j'}\big]_{r_1,r_2} \eq \big[(a^{\dag}a)_{q,j}\big]_{r_1,r_2} + 2(j-j')\delta_{r_1,r_2}\,.
 \end{split}
 \label{eq:adaga'_matrix}
\end{align}

In conclusion, in contrast to $\HTC$, the matrix elements of Hamiltonians $J_z$ and $a^\dag a$ do change under this transformation, but only by the very particular form of shifts by $j-j'$.
As we show in the following section, this observation, together with the invariance of $\HTC$, imposes interesting constraints on the dynamics of systems evolving under Hamiltonians that can be written as a linear combination of these three Hamiltonians, which, for instance, explains the conservation law observed in \cref{fig:cons_law}.

Finally, it is worth emphasizing that a similar accidental symmetry exists in the so-called anti–Tavis–Cummings Hamiltonian, obtained by switching $J_+$ and $J_-$ in \cref{eq:TCham}, or equivalently by exchanging $a$ and $a^\dag$,  given by 
\begin{align}
    H_{\text{anti-TC}}=g_{\text{anti-TC}} (J_+ a^\dag+J_- a) = e^{i \pi J_x} \HTC e^{-i \pi J_x} \,. 
\end{align}
In particular, under a $\pi$-rotation about the $x$-axis, the state $\ket{j,m,\al}$ transforms as $e^{-i\pi J_x}\ket{j,m,\al}=(-1)^{j-m}\ket{j,-m,\al}$. 
Therefore, the Hamiltonian remains invariant under the transformation in \cref{eq:trans1}, provided $m$ and $m'$ are simultaneously replaced by $-m$ and $-m'$. 
The overall representation of the accidental symmetry for anti-Tavis-Cummings interaction is thus
\begin{align}
 \begin{split}
    \begin{pmatrix}
        j \\
        m \\
        k
    \end{pmatrix}\longrightarrow 
    \begin{pmatrix}
        j' \\
        m' \\
        k'
    \end{pmatrix}
    =
    \begin{pmatrix}
        \sfrac{1}{2} & -\sfrac{1}{2} & \sfrac{1}{2} \\
        -\sfrac{1}{2} & \sfrac{1}{2} & \sfrac{1}{2} \\
        1 & 1 & 0
    \end{pmatrix}
    \begin{pmatrix}
        j \\
        m \\
        k
    \end{pmatrix}\,.
 \end{split}
\end{align}

\section{A Conserved Observable}
\label{sec:conserved_observable}
Next, we introduce operators and conserved observables that make the consequences of this symmetry explicit.
For each pair $(j,j')$ in the range
\[ \jmin\,\leq\, j'\,<\,j\,\leq\,\frac{n}{2}\,, \]
we define the operator $S_{j,j';\,\alpha,\alpha'}$ that maps states from sectors with angular momenta $j$ to $j'$, while keeping their $m$ order, such that
\begin{align}\label{eq:states}
    S_{j,j';\,\alpha,\alpha'} \big|r^{(q,j,\alpha)}\big\rangle \eq \big|r^{(q',j',\alpha')}\big\rangle\,\,:\,\, r=0,\,\cdots,\,2j',
\end{align}
\noindent where $\alpha$ and $\alpha'$ are arbitrary, but fixed multiplicity indices for angular momenta $j$ and $j'$, respectively.
(See the examples in \cref{fig:pairwise_equiv} and \cref{fig:accidental_3qubit}.) 
Equivalently,
\begin{align} \label{eq:Sjj}
    S_{j,j';\,\alpha,\alpha'} \eq \sum_{r=0}^{2j'} \big|r^{(q',j',\alpha')}\big\rangle\big\langle r^{(q,j,\alpha)}\big|\,.
\end{align}

In the rest of the paper, since the multiplicity indices $\alpha$ and $\alpha'$ are fixed, but arbitrary, we suppress them from the notation and denote this operator as $S_{j,j'}$.
It is also worth noting that $S_{j,j'}$ is a partial isometry: it preserves inner products on the orthogonal complement of its kernel, i.e., the subspace spanned by the states on the left-hand sides of \cref{eq:states}.

Then, \cref{eq:HTCphi_matrix} implies that $\HTC(\phi)$ satisfies
\begin{align}\label{eq:sds}
    \big[\HTC(\phi) , S_{j,j'}\big]=0\,,
\end{align}
for all $\phi\in[0,2\pi)$.
Furthermore, \cref{eq:Jz'_matrix} and \cref{eq:adaga'_matrix} can be rewritten as
\begin{align}
 \begin{split}
    [J_z ,  S_{j,j'}]&= (j-j') S_{j,j'} \\[4pt] 
    [a^\dag a ,  S_{j,j'}]&= 2(j-j') S_{j,j'}\,.
 \end{split}
\end{align}
Therefore, Hamiltonian $H(t)$ in \cref{eq:general_Ham} satisfies
\begin{align}
    \big[H(t), S_{j,j'}\big]= h(t) S_{j,j'}\,,
\end{align}
for a real function $h(t)=(j-j')[2\omega_{\text{osc}}(t)+\omega_{z}(t)]$.

We emphasize that, strictly speaking, the Hamiltonians $J_z$ and $a^\dag a$ break this symmetry.
Nevertheless, the breaking occurs in a highly constrained manner, as identified in the equation above, and this constraint leads to important restrictions on the resulting dynamics. 
That is, all unitary transformations $V(t)$ realized by Hamiltonians  $\HTC$, $J_z$, and $a^\dag a$ under the Schr\"odinger equation
\begin{align}
    \frac{d}{dt}V(t)=-i H(t) V(t)\,,
\end{align}
with $V(0)=\mathbb{I}$, satisfy a restriction identified by the following statement.
\begin{proposition}
For $t\ge 0$, let $V(t)$ denote the unitary evolution realized by a  Hamiltonian $H(t)$ via the Schr\"odinger equation from time $0$ to $t$.
Suppose there exists an operator $A$ such that, for all $t$, $[H(t),A]=h(t)\,A,$ where $h(t)$ is a real-valued function. 
Then
\begin{align}\label{eq:sjs}
    V(t)\,A = e^{i\phi(t)}\, A\, V(t)\,,
\end{align}
where the phase is given by $\phi(t)=-\int_0^t h(s)\,ds$.
\end{proposition}
\begin{proof}
Define operator $V^\dag(t) A V(t) e^{i\phi(t)}$, where $\phi(t)=-\int_0^t h(s) ds$. 
Then,
\begin{align}
 \begin{split}
    &\frac{d}{dt} \Big\{e^{i\phi(t)} V^\dag(t) A V(t)\Big\}\\[8pt]
    &\hspace{4pt}=e^{i\phi(t)} V^\dag(t) \big( -i[A, H]-i A h(t)  \big)V(t) =0\,.
 \end{split}
\end{align}
Therefore,
\begin{align}
    e^{i\phi(t)} V^\dag(t) A V(t)=e^{i\phi(0)} V^\dag(0) A V(0)=A\,
\end{align}
which implies \cref{eq:sjs}.
\end{proof}
Choosing $A=S_{j,j'}$ we find that any unitary $V$ realized by Hamiltonian $H(t)$ in \cref{eq:general_Ham} satisfies
\begin{align}\label{eq:unit} 
    V S_{j,j'}= e^{i\phi}  S_{j,j'} V \,,
\end{align}
for some phase $e^{i\phi}$. 
In words, this means that any unitary transformation that can be realized using Hamiltonians $\HTC$, $a^\dag a$, and $J_z$ acts identically, up to a relative phase, on the two subspaces corresponding to the left- and right-hand sides of \cref{eq:states}.
See also \cref{sec:controllability}.

\subsection{Corresponding Conservation Laws}
\label{sec:conserved_S}
\Cref{eq:unit} is equivalent to $V^\dag S_{j,j'} V= e^{-i\phi} S_{j,j'}$ which in turn, means for any initial joint state $\ket{\psi_1}$ of the qubits and bosonic mode, and the corresponding final state $\ket{\psi_2}=V\ket{\psi_1}$, we obtain the conservation law
\begin{align}
    \Big|\ipo{\psi_2}{S_{j,j'}}{\psi_2}\Big| \eq \Big|\ipo{\psi_1}{S_{j,j'}}{\psi_1}\Big|\,.
 \label{eq:cons_law}
\end{align}
Crucially, while any unitary $V$ realizable with the Hamiltonian $H(t)$ is PI and commutes with $Q$, generic unitaries that satisfy these properties do not necessarily obey the above conservation law. 
For instance, for the family of unitaries $\exp(-i t J_z^2)$, this conservation law is violated (see \cref{fig:cons_law}).\footnote{Note that in general the expectation value of $S_{j,j'}$, which is a complex number, is not conserved, but its absolute value remains conserved.}

\vspace{3mm}
We finish this section with some remarks regarding the definition of the operator $S_{j,j'}$. 
First, $S_{j,j'}$ breaks both permutational and SU(2) symmetry, because it connects sectors with different angular momenta $j$ and $j'$.
Second, the definition of $S_{j,j'}$ in \cref{eq:Sjj} depends on the choice of the suppressed multiplicity indices $\alpha$ and $\alpha'$ in the states $\ket{j,m,\alpha}$ and $\ket{j',m',\alpha'}$. 
For any other choices $\tilde{\alpha}$ and $\tilde{\alpha}'$, the corresponding operator $\widetilde{S}_{j,j'}$ can be written as
\begin{align}
    \widetilde{S}_{j,j'} = T_2 S_{j,j'} T_1 \,,
\end{align}
where $T_1$ and $T_2$ are unitary transformations that can be expressed as linear combinations of permutations.

It follows that for any PI operator $B$, such as $\HTC(\phi)$, the following statements are equivalent:
\begin{align}
    \big[B,\,\widetilde{S}_{j,j'}\big] = 0 \,\,\, \Longleftrightarrow\,\,\,\big[B,\,S_{j,j'}\big] = 0 \,.
\end{align}
Therefore, for PI operators, conservation of $\widetilde{S}_{j,j'}$ does not impose any additional constraints beyond those implied by conservation of $S_{j,j'}$, so fixing $\al$ does not hide any features of the accidental symmetry.

Finally, we note that while $S_{j,j'}$ is not Hermitian, we can construct Hermitian operators that commute with $\HTC$. 
In particular, \cref{eq:sds} implies that observable
\begin{align} \label{eq:S}
    S=\sum_{j=j_{\text{min}}+1}^{n/2}\,\sum_{j'=\jmin}^{j-1} S_{j,j'}+S_{j,j'}^\dag\,,
\end{align}
satisfies
\begin{align} \label{eq:comm0}
    \big[\HTC(\phi),\,S\big]=0\,,
\end{align}
for all $\phi\in[0,2\pi)$. 
It is also worth noting that \cref{eq:sds} could also be obtained from \cref{eq:comm0}, using the fact that $\HTC(\phi)$ has also permutational symmetry. 
In particular, $S_{j,j'}=\Pi_{j'} S \Pi_{j}$, where $\Pi_j$ is the projector onto the eigensubspace of $J^2$ with eigenvalue $j(j+1)$. 
Recall that $S_{j,j'}$ acts on subspaces with fixed multiplicity indices $\al,\al'$.
Since $\Pi_j$ can be written as a linear combination of permutations, it follows that $[\HTC, S_{j,j'}] = 0$ can be deduced from $[\HTC, S] = 0$, together with the permutational symmetry of $\HTC$.

\section{Characterizing symmetry-respecting operators}
\subsection{Permutation-invariant (PI), U(1)-invariant Operators} \label{sec:sector_decomp}
To better understand the constraints imposed by the accidental symmetry, it is useful to characterize the set of all operators on the total Hilbert space that respect both the permutational symmetry and U(1) symmetry with corresponding charge $Q$.

Recall that in \cref{sec:invariant}, we defined invariant subspaces $\H\qj$ of the combined qubit-bosonic Hilbert space,
\begin{align}
    \H_{\text{qubits}}\otimes\H_{\text{osc}} = (\C^2)^{\otimes n}\otimes\mathcal{L}^2(\R)\,,
 \label{eq:Hilbert}
\end{align}
with fixed charge $q$ and angular momentum $j$.
In an $n$-qubit system, the angular momentum $j$ subspace appears with multiplicity $M(n,j)$, which means that the subspace with charge $q$ and angular momentum $j$ decomposes as
\begin{align}
    \C^{M(n,j)} \otimes\H\qj\,.
\end{align}
Thus, the overall Hilbert space decomposes as
\begin{align}\label{eq:decomp}
    \H_{\text{qubits}}\otimes\H_{\text{osc}}\,=\,\bigoplus_{q=0}^{\infty}\, \bigoplus_{j=\max(\jmin,\,\frac{n}{2}-q)}^{n/2}\left(\C^{M(n,j)} \otimes\H\qj\right)\,, 
\end{align}
where for each $q$, the minimum value of $j$ is $\max(\jmin,\,\frac{n}{2}-q)$.

Finally, we note that permutations of qubits conserve both the charge operator $Q$ and $J^2$, and therefore are block-diagonal with respect to this decomposition. 
Furthermore, permutations act only on the multiplicity spaces $\mathbb{C}^{M(n,j)}$. 
Moreover, by Schur–Weyl duality, permutations act irreducibly on these spaces, with each value of $j$ corresponding to an inequivalent representation of the permutation group $\Sn$.
This, in particular, implies:
\begin{proposition} \label{prop:components}
Consider an arbitrary operator $A$ on $\H_{\text{qubits}}\otimes\H_{\text{osc}}$ commuting with the charge operator $Q=J_z+a^\dag a+n/2$.
Then, $A$ is PI if, and only if, with respect to decomposition in \cref{eq:decomp}, it takes the form
\begin{align}\label{eq:perm}
   A\,=\,\bigoplus_{q=0}^{\infty}\, \bigoplus_{j=\max(\jmin,\,\frac{n}{2}-q)}^{n/2}\left(\mathbb{I}_{M(n,j)} \otimes A_{q,j}\right)\,, 
\end{align}
where $\mathbb{I}_{M(n,j)}$ is the identity operator on $\C^{M(n,j)}$, and $A_{q,j}$ is the component of $A$ in $\mathcal{H}_{q,j}$.
\end{proposition}
This follows from the following facts: (i) $[A,Q]=0$ implies that $A$ is block-diagonal with respect to sectors with different charge $q$; (ii) The group of permutations $\mathbb{S}_n$ acts trivially on $\mathcal{H}_{q,j}$ and irreducibly on $\mathbb{C}^{M(n,j)}$. 
Then, applying Schur's lemma proves the claim.

\subsection{Identifying Correlated Sectors}
\label{sec:correlated_sectors}
For a general PI, U(1)-invariant operator $A$, the operators $A_{q,j}$ in \cref{eq:perm} can be independent of each other; for instance, all can be zero except for $A_{q,j}$ corresponding to a particular choice of $q$ and $j$.
However, suppose $A$ also respects the accidental symmetry, such that $[A,S_{j,j'}]=0$ (equivalently, $AS_{j,j'}=S_{j,j'}A$) for all $(j,j')$ in the range $\jmin\leq j'<j\leq n/2$.
This immediately implies that the matrix elements of $A$ satisfy, using the notation introduced in \cref{eq:r_matrix_els},
\footnote{\begin{align*}
    [A_{q,j}]_{r_1,r_2} &\eq \big\langle r_1^{(q',j',\alpha')}\big|\,S_{j,j'}A\,\big|r_2^{(q,j,\alpha)}\big\rangle\\[6pt]
    &\eq \big\langle r_1^{(q',j',\alpha')}\big|\,AS_{j,j'}\,\big|r_2^{(q,j,\alpha)}\big\rangle\\[6pt]
    &\eq \big\langle r_1^{(q',j',\alpha')}\big|\,A\,\big|r_2^{(q',j',\alpha')}\big\rangle \hspace{16pt}\eq [A_{q',j'}]_{r_1,r_2}\,.\end{align*}
}
\begin{align}\label{eq10}
    [A_{q,j}]_{r_1,r_2} = [A_{q',j'}]_{r_1,r_2} \,\,:\,\, r_1,r_2=0, \dotsc, 2j'\,,
\end{align}
where $q$ and $q'$ are given by \cref{eq:qq'_matrix}.
In other words, the components of $A$ in sectors $\H\qj$ and $\H_{q',j'}$ are unitarily equivalent.
Indeed, with respect to the basis in \cref{eq:r_qja}, ordered by increasing $r$, their matrices are equal.
The Tavis-Cummings Hamiltonian is of course, such an operator: see \cref{prop:main_prop}.

\subsection{Dynamical constraints imposed by the accidental symmetry}
\label{sec:controllability}
Consider a unitary operator $V$ that can be realized using Hamiltonians $\HTC$, $J_z$, and $a^\dag a$; that is, Hamiltonians $H(t)$ of the form given in \cref{eq:general_Ham}.
The permutational and U(1) symmetries of $H(t)$ imply that $V$ inherits the same symmetries.
Therefore, by \Cref{prop:components}, it decomposes as
\begin{align}
 \label{eq:Vform}
    V=\bigoplus_{q=0}^{\infty}\, \bigoplus_{j=\max(\jmin,\,\frac{n}{2}-q)}^{n/2}\left(\mathbb{I}_{M(n,j)} \otimes v_{q,j}\right)\,.
\end{align}
For a general PI, U(1)-invariant unitary $V$, the components $v_{q,j}$ are arbitrary and can be chosen independently; that is, they do not satisfy any constraints among themselves.
On the other hand, the accidental symmetry implies that certain pairs of unitaries $v_{q,j}$ must be equal up to an unspecified phase.
In particular, as shown in \cref{eq:unit}, any unitary $V$ realizable by a Hamiltonian $H(t)$ satisfies
\begin{align}
    V S_{j,j'} = e^{i\phi} S_{j,j'} V\quad:\quad 0\leq j' < j\leq \frac{n}{2}\,.
\end{align}
Then, similarly to \cref{eq10}, we find
\begin{align}
    \big[v_{q',j'}\big]_{r_1,r_2} = e^{i\phi} \big[v_{q,j}\big]_{r_1,r_2}\,\,:\,\, r_1,r_2=0, \dotsc, 2j'\,.
\end{align}
In other words, the components of $V$ in sectors $\H\qj$ and $\H_{q',j'}$ are equal, up to a fixed relative phase.

In our companion paper \cite{theory_paper}, we fully characterize all unitaries $V$ that can be realized using Hamiltonians $H(t)$ in \cref{eq:general_Ham}, i.e., using the three interactions $\HTC$, $J_z$, and $a^\dag a$. 
Remarkably, we find that any PI, U(1)-invariant unitary is realizable, provided that (i) it respects the constraint imposed by the above accidental symmetry, and (ii) the determinants of the unitaries $v_{q,j}$ satisfy certain constraints among themselves, which are related to constraints on the center of the Lie algebra of realizable Hamiltonians \cite{PhysRevA.92.042309, marvian_sym_loc_2022}.
Therefore, up to the latter constraints, which typically arise in the presence of symmetry and locality \cite{marvian_sym_loc_2022,HLM_2024_SU(d)}, the constraints imposed by the accidental symmetry are the only additional restrictions. 
In \cite{theory_paper} we also show that adding Hamiltonian $J_z^2$ is sufficient to remove the constraints imposed by the accidental symmetry.

\section{Simplest non-trivial example: 3 qubits}
\label{sec:3Q_example}
First, we briefly comment on the case of $n=2$ qubits, in which case the map in \cref{eq:states} yields exactly one non-trivial equivalence, namely
\begin{align}
 \begin{split}
   S_{2,0}\big(&\ket{j=1,m=-1}\otimes\vac \big)\\[4pt]
   &\hspace{12pt}\eq \ket{j=0,m=0}\otimes\ket{2}_{\text{osc}}\,,
 \end{split}
\end{align}
where
\begin{align}
    S_{2,0} = \ket{j=0,m=0}\bra{j=1,m=-1} \otimes \ket{2}\bra{0}_{\rm osc}\,,
\end{align}
which relates 1D sectors $\mathcal{H}_{q,j}$ with $j=1$ and $q=0$ and $\mathcal{H}_{q',j'}$ with $j'=0$ and $q'=2$.
The condition $[S_{2,0},\HTC]=0$ implies that the two states related by $S_{2,0}$ must be eigenvectors of $\HTC$ with the \emph{same eigenvalue}.
Indeed, both states have eigenvalue zero under $\HTC$.
Note that, as discussed in \cref{sec:1D} and \cref{fig:2q_inv_subs}, this is a genuine constraint beyond permutation invariance and the U(1)
symmetry.

\vspace{0.5em}
The first genuinely interesting and less obvious example of the accidental
symmetry occurs for $n=3$ qubits.
Recall that for a $3$-qubit system, $j$ takes the values $3/2$ and $1/2$, where the irreducible representation with angular momentum $j=3/2$ is multiplicity-free, whereas the representation with $j=1/2$ appears with multiplicity two.

Using the $\ket{j,m}$ basis (suppressing multiplicity index $\alpha$) to represent the qubit states, \cref{eq:states} yields the following nontrivial maps:
\begin{align}\label{eq:3q_equiv}
 \begin{split}
    S_{3/2,1/2}(\ket{\sfrac{3}{2},-\sfrac{1}{2}}\otimes\vac) &= \ket{\sfrac{1}{2},\sfrac{1}{2}}\otimes\ket{2}_{\text{osc}}\\[4pt]
    S_{3/2,1/2}(\ket{\sfrac{3}{2},-\sfrac{3}{2}}\otimes\ket{1}_{\text{osc}}) &=
    \ket{\sfrac{1}{2},-\sfrac{1}{2}}\otimes\ket{3}_{\text{osc}}
 \end{split}\,.
\end{align}
Since $j=1/2$ appears with multiplicity two, there is a freedom in choosing state $\ket{\sfrac{1}{2},\pm\sfrac{1}{2}}$. 
A convenient choice is to follow the convention in \cref{eq:ex}, which means
\begin{align}\label{eq:Jz_elements_from_map}
\begin{split}
    \ket{\sfrac{1}{2},\sfrac{1}{2},\al_0} &:= \ket{\Psi^-}_{12}\ket{\uparrow}_3 \\[6pt]
    \ket{\sfrac{1}{2},-\sfrac{1}{2},\al_0} &:= \ket{\Psi^-}_{12}\ket{\downarrow}_3\,,
 \end{split}
\end{align}
as the two basis states with $j=1/2$ and $m=\pm 1/2$, where $\ket{\uparrow}$, $\ket{\downarrow}$ are the $+1$ and $-1$ eigenstates of $\sigma_z$, respectively.
Since $j=1/2$ has multiplicity two, instead of states $\ket{\sfrac{1}{2},\sfrac{1}{2},\al_0}$ in the above constructions, one could also consider, e.g., the orthogonal states
\begin{align}
 \begin{split}
    \ket{\sfrac{1}{2},\sfrac{1}{2},\al_1} &:= \frac{\ket{\uparrow}_1\ket{\Psi^-}_{23}+\ket{\uparrow}_2\ket{\Psi^-}_{13}}{\sqrt{2}}\\[6pt]
    \ket{\sfrac{1}{2},-\sfrac{1}{2},\al_1} &:= \frac{\ket{\downarrow}_1\ket{\Psi^-}_{23}+\ket{\downarrow}_2\ket{\Psi^-}_{13}}{\sqrt{2}}\,.
 \end{split}
\end{align}
Then, the accidental symmetry implies that $\HTC$ acts identically within the 2D subspace
\begin{align}
 \begin{split}
    \H_{q=1,\,j=3/2} \,:=\, \s_{\mathbb{C}}\big\{&\ket{\sfrac{3}{2},-\sfrac{1}{2}}\otimes\vac, \\
    &\ket{\sfrac{3}{2},-\sfrac{3}{2}}\otimes\ket{1}_{\text{osc}}\big\}
 \end{split}
 \label{eq:Hqj_3qubit_1}
\end{align}
which contains \textit{all} states with angular momentum $j=3/2$ and charge $q=1$, and
\begin{align}
 \begin{split}
    \H_{q'=4,\,j'=1/2} \,:=\, \s_{\mathbb{C}}\big\{&\ket{\sfrac{1}{2},\sfrac{1}{2}}\otimes\ket{2}_{\text{osc}}, \\
    &\ket{\sfrac{1}{2},-\sfrac{1}{2}}\otimes\ket{3}_{\text{osc}}\big\}
 \end{split}
 \label{eq:Hqj_3qubit_2}
\end{align}
for a fixed multiplicity index $\alpha$, e.g., $\alpha_0$ or $\alpha_1$, which contains \textit{all} states with angular momentum $j=1/2$ and charge $q=4$ for that fixed choice of $\al$. 
Explicitly, the only non-zero matrix elements of $\HTC$ in these subspaces are
\begin{align*}
    \sqrt{3} &= (\bra{\tfrac{3}{2},-\sfrac{1}{2}}\otimes \bra{0}_{\text{osc}})\, \HTC \,(\ket{\sfrac{3}{2},-\sfrac{3}{2}}\otimes\ket{1}_{\text{osc}})\\[4pt]
    &=(\bra{\sfrac{1}{2},\sfrac{1}{2},\al_0}\otimes \bra{2}_{\text{osc}})\, \HTC \,(\ket{\sfrac{1}{2},-\sfrac{1}{2},\al_0}\otimes\ket{3}_{\text{osc}})\\[4pt]
    &=(\bra{\sfrac{1}{2},\sfrac{1}{2},\al_1}\otimes \bra{2}_{\text{osc}})\, \HTC \,(\ket{\sfrac{1}{2},-\sfrac{1}{2},\al_1}\otimes\ket{3}_{\text{osc}})\,.
\end{align*}
\Cref{fig:accidental_3qubit} illustrates this phenomenon.

Also, recall from the previous section that $J_z$ is not quite identical in the two sectors defined in \Cref{eq:Hqj_3qubit_1,eq:Hqj_3qubit_2} \Crefrange{eq:Hqj_3qubit_1}{eq:Hqj_3qubit_2} -- rather its matrices differ by $(j-j')\1=(\sfrac{3}{2}-\sfrac{1}{2})\1=\1$, for the 3-qubit case.
Indeed, with respect to the basis introduced in \Cref{eq:r_qja}, ordered by increasing $r$ (equivalently, $m$),
\begin{align}
    \big[(J_z)_{q=1,\,j=3/2}\big] = \big[(J_z)_{q'=4,\,j'=1/2}\big] - \1\,.
\end{align}
For the case of $n=3$ qubits discussed in this section, \cref{fig:cons_law} illustrates examples of how the conservation law in \cref{eq:cons_law} holds for Hamiltonians $\HTC$, $J_z$, and $a^{\dag}a$ -- but is violated by the Hamiltonian $J_z^2$.
In particular, the absolute value of the expectation value of operator $S(\sfrac{3}{2},\sfrac{1}{2})$ as a function of time, with respect to each of these four Hamiltonians, is illustrated for two different initial states.
Each initial state $\ket{\psi_0}$ has components in both accidental symmetry sectors $\H_{q=1,j=3/2}$ and $\H_{q'=4,j'=1/2}$, such that there is nonzero overlap $\ipo{\psi_0}{S(\sfrac{3}{2},\sfrac{1}{2})}{\psi_0}\neq0$.
\begin{figure}[htp]
    \centering
    \includegraphics[width=0.9\linewidth]{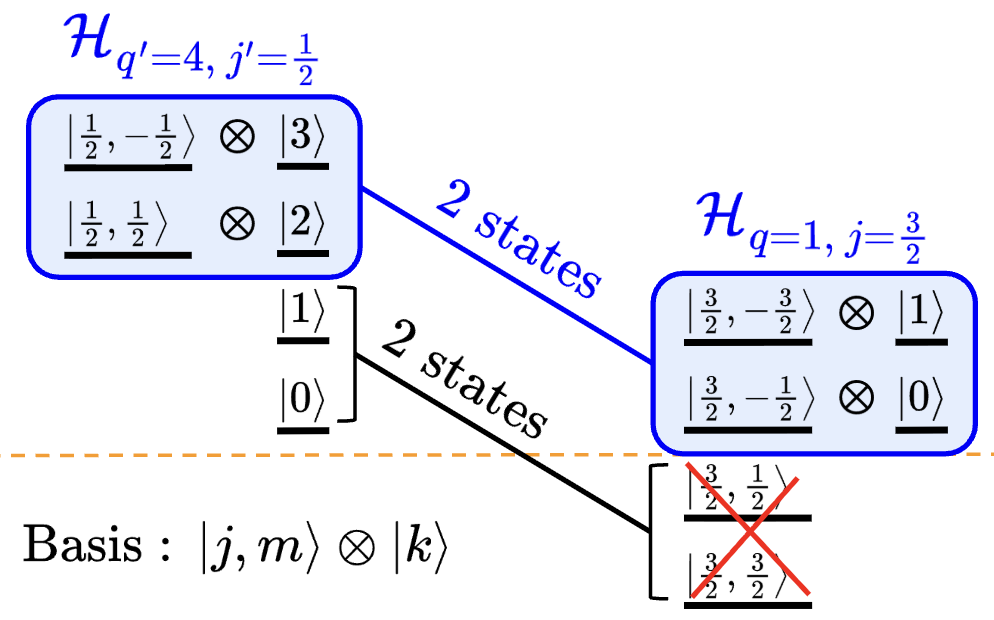}
    \caption{\textbf{3-qubit Example:} two sectors related by accidental symmetry}
    \label{fig:accidental_3qubit}
\end{figure}
\begin{figure*}[htp]
    \centering
    \includegraphics[width=0.99\textwidth]{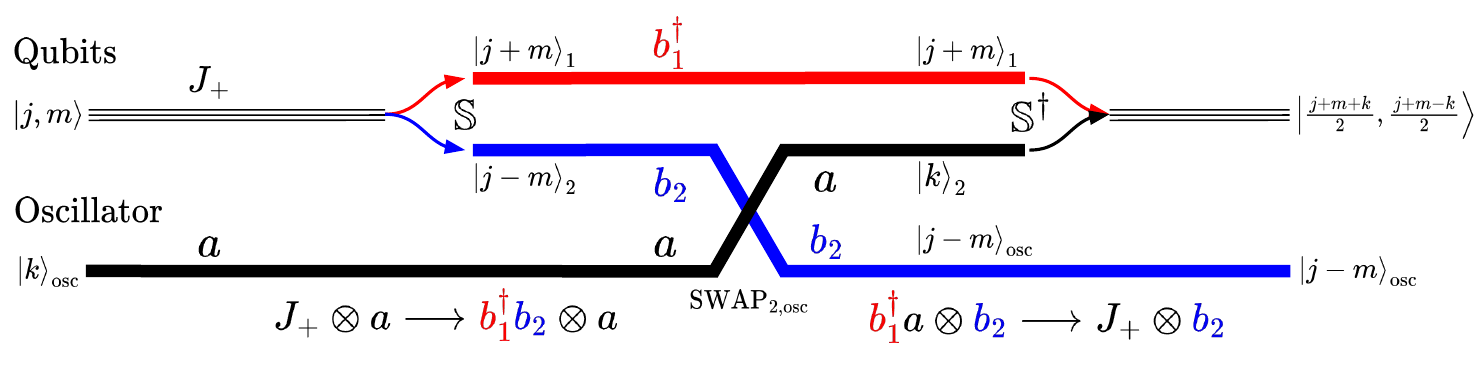}
    \caption{\textbf{Explaining the accidental symmetry of Tavis-Cummings interaction via Schwinger's model of angular momentum.} Illustration of the symmetry transformation $W=\mathbb{S}^{\dag}\text{SWAP}_{2,\text{osc}}\mathbb{S}$. 
    See below for further details.}
    \label{fig:schwinger}
\end{figure*}

\section{Explaining the accidental symmetry using Schwinger's oscillator representation of angular momentum}\label{sec:schwinger}
Here, we further discuss the accidental symmetry and explain how it can be understood via Schwinger's representation of angular momentum.
It is worth mentioning that the Schwinger representation has been used to study the TC model before, e.g., in \cite{Vadeiko_etal_2003}, albeit not in connection with the accidental symmetry identified here.

Introduced by Julian Schwinger in the 1950s \cite{Schwinger_model_1952}, this representation expresses spin operators in terms of two independent quantum harmonic oscillators, also called bosonic modes.
This construction has since become a standard tool in quantum optics and many-body physics due to its algebraic simplicity and versatility.
Consider a pair of uncoupled harmonic oscillators with annihilation operators $b_1$ and $b_2$, satisfying
\begin{align}
    [b_1, b_1^\dag]=[b_2, b_2^\dag]=\mathbb{I}\,,
\end{align}
and
\begin{align}
    [b_1,b_2]=[b_1,b^\dag_2]=0\,.
\end{align}
For integers $r_1,r_2\ge 0$, let $\ket{r_1}_1\ket{r_2}_2$ be the joint state of these two oscillators, in which oscillators 1 and 2 have $r_1$ and $r_2$ excitations, respectively. 
That is, it is an eigenstate of $b^\dag_1b_1$ with eigenvalue $r_1$ and eigenstate of $b^\dag_2b_2$ with eigenvalue $r_2$. 

Now consider the irrep of SU(2) with angular momentum $j$. 
According to Schwinger's description, the corresponding $2j+1$ states $\ket{j,m}: m=-j,\,\dots,\,j$ can be identified with $2j+1$ states of a pair of oscillators with a total of $2j$ excitations.
That is,
\begin{align}\label{eq:Sch}
    \mathbb{S}:\,\ket{j,m} \longrightarrow \ket{j+m}_1\otimes\ket{j-m}_2\,.
\end{align}
Technically, this map defines a unitary transformation
\begin{align}
    \mathbb{S}:\bigoplus_{d=0}^\infty \mathbb{C}^{d} \rightarrow \mathcal{L}^2(\mathbb{R})\otimes \mathcal{L}^2(\mathbb{R})
\end{align}
embedding all irreps of SU(2) with dimension $d=2j+1$ in the Hilbert space of two oscillators.
Note that this map preserves inner products, and is surjective, i.e., for any integers $k_1, k_2 \ge 0$,
\begin{align}
    \mathbb{S}^\dag \big(\ket{k_1}_1\otimes \ket{k_2}_2\big)= \left|j=\tfrac{k_1+k_2}{2}, \,m=\tfrac{k_1-k_2}{2}\right\rangle\,,
\end{align}
where $j=(k_1+k_2)/2$ and $m=(k_1-k_2)/2$ are either both integers or half-integers, and $|m|\le j$, which means they correspond to a valid element of the basis $\{\ket{j,m}\}$. 

It can be easily seen that under this transformation, the angular momentum operators transform as
\bes
\begin{align}
    \mathbb{S} J_z \mathbb{S}^\dag &=  \frac{1}{2}\big(b^\dag_1b_1-b^\dag_2b_2\big)\,,\\
    \mathbb{S}  J_+ \mathbb{S}^\dag&=  b_1^\dag b_2\,.
\end{align}
\ees
Note that, unlike the original definition in \cref{eq:TCham}, we now interpret the angular momentum operators $J_x$, $J_y$, and $J_z$ as operators acting on the space $\bigoplus_{d=0}^\infty \mathbb{C}^{d}$, which realizes all irreps of $\mathrm{SU}(2)$ exactly once, with no multiplicities.
This is distinct from the original representation on a system of $n$ qubits, where certain angular momentum sectors occur with nontrivial multiplicities. 
However, since the accidental symmetry concerns the matrix elements of $\HTC$, which are independent of the multiplicity labels, this distinction is irrelevant for our purposes.

\subsection{Exchange Symmetry as the Origin of the Accidental Symmetry}
Now consider operator $J_+\otimes a$, where $a$ is the bosonic annihilation operator, as defined before (see \cref{fig:schwinger}).
Applying the above transformation, this operator will be mapped to
\begin{align}
    \mathbb{S}( J_+\otimes a)\mathbb{S}^\dag= b_1^\dag b_2a\,,
\end{align}
which means $\HTC(\phi)=\gTC(e^{i\phi}J_+\otimes a+e^{-i\phi} J_-\otimes a^\dag)$ is mapped to
\begin{align}
    \HTC(\phi)\,\longrightarrow\,  \gTC( e^{i\phi} b_1^\dag b_2a+e^{-i\phi} b_1b_2^{\dag}a^{\dag})\,.
\end{align}
Here is the key observation: operator $b_1^\dag b_2a$ remains invariant under swapping the real oscillator with annihilation operator $a$ and the second virtual Schwinger's oscillator with annihilation operator $b_2$. 
In other words,
\begin{align}
    \text{SWAP}_{2,\text{osc}} (b_1^\dag b_2  a) \text{SWAP}_{2,\text{osc}} = b_1^\dag b_2  a\,.
\end{align}
Therefore, going back to the space $\bigoplus_{d=0}^\infty \mathbb{C}^{d}$, we find that operator $J_+\otimes a $ satisfies
\begin{align}
    W^\dag (J_+\otimes a)  W = J_+\otimes a \,,
\end{align}
where
\begin{align}
    W=\mathbb{S}^\dag\,\text{SWAP}_{2,\text{osc}}\,\mathbb{S}\,, 
\end{align}
and $W^2=\mathbb{I}$. 
Note that, as mentioned above, we are interpreting the angular momentum operators $J_\pm$ on the infinite-dimensional space $\bigoplus_{d=0}^\infty \mathbb{C}^{d}$.
In terms of basis $\ket{j,m}\ket{k}_{\text{osc}}$, operator $W$ acts as
\begin{align}\label{eq:travs}
 \begin{split}
    \textstyle W \ket{j,m}\ket{k}_{\text{osc}} &= \mathbb{S}^\dag\,\text{SWAP}_{2,\text{osc}}\big(\ket{j+m}_1 \ket{j-m}_2 
    \ket{k}_{\text{osc}}\big)\\[4pt]
    &=\mathbb{S}^\dag\,\big(\ket{j+m}_1 \ket{k}_2
    \ket{j-m}_{\text{osc}}\big)\\[4pt]
    &=
    \left|\frac{j+m+k}{2},\frac{j+m-k}{2}\right\rangle\ket{j-m}_{\text{osc}}\\[4pt] &=\ket{j',m'}\ket{k'}_{\text{osc}}\,,
 \end{split}
\end{align}
where $j',m', k'$ are exactly as given in \cref{eq:trans1}. 
Furthermore,
\begin{align}
    r \eq j+m \eq j'+m'\,,
\end{align}
is the number of excitations in the first Schwinger oscillator, corresponding to the creation operator $b_1^\dag$, which remains conserved because $\text{SWAP}_{2,\text{osc}}$ acts trivially on this system.
Also, note that the action of the map $W$, when restricted to two sectors $j$ and $j'$, is equivalent to the map $S_{j,j'}$ defined in \cref{eq:Sjj}.
\Cref{fig:schwinger} illustrates the overall effect of this process on operator $J_+\otimes a$ and states $\ket{j,m}\otimes\ket{k}_{\text{osc}}$.

Hence, in the Schwinger picture, we find an explanation for the accidental symmetry of $\HTC(\phi)$, namely, that $\HTC(\phi)$ treats the physical oscillator and the second virtual oscillator equivalently.
Note that a similar relation holds for the anti-TC interaction, $e^{i\phi} J_+ a^\dag + e^{-i\phi} J_- a$.
The only difference is that, in that case, instead of swapping the oscillator with the second virtual oscillator (highlighted in blue in \cref{fig:schwinger}), it should be swapped with  the first virtual oscillator (highlighted in red in \cref{fig:schwinger}).

Finally, note also that operator $W$ transforms $J_z$ as
\begin{align}
 \begin{split}
    J_z'\,&:=\,W^{\dag}J_zW \\
    &=\, \frac{1}{2}\Big(\sqrt{J^2+\sfrac{1}{4}}-\sfrac{1}{2}) + J_z -(a^{\dag}a)_{\text{osc}}\Big)\,,
 \end{split}
\end{align}
which follows from \cref{eq:travs}, and agrees with what we expect from the map in \cref{eq:trans1}, namely $m'=(j+m-k)/2$.

\section{Summary \& Discussion}
In this paper, we identified an ``accidental'' symmetry of the Tavis-Cummings Hamiltonian (\cref{eq:TCham}), which is independent of its ``standard'' permutational and U(1) symmetries.
We first discussed this symmetry in terms of the map in \cref{eq:trans1}, under which certain matrix elements of $\HTC$ remain unchanged.
Equivalently, we formulated the accidental symmetry in terms of a conserved observable $S$ (see \cref{sec:conserved_S}).
Also, we showed how this symmetry can be understood using Schwinger's oscillator model of angular momentum.
Finally, we briefly discussed implications for quantum computing -- in particular how controllability of qubit-bosonic systems interacting via the TC Hamiltonian is restricted by the accidental symmetry.

In \cite{theory_paper}, we investigate in detail the implications of this accidental symmetry for quantum control.
In particular, we fully characterize the sets of unitaries realizable using $\HTC$ together with permutation-invariant (PI) qubit Hamiltonians, such as $J_z$.
We also determine which PI gates on the qubits can be implemented by using the bosonic mode as an ancillary system, and by allowing Hamiltonian $J_x$ in addition to $J_z$ and $\HTC$.
Also, in \cite{circuit_paper}, we provide explicit quantum circuit implementations for all 2-qubit PI unitaries using $\HTC$, $J_z$, and $J_x$. \\

\section*{Acknowledgments}
This work was supported in part by a collaboration between the U.S. Department of Energy and other agencies.
This material is based upon work supported by the U.S. Department of Energy, Office of Science, National Quantum Information Science Research Centers, Quantum Systems Accelerator (Award No. DE-SCL0000121). 
Additional support is acknowledged from NSF PHY-2046195, NSF QLCI grant OMA-2120757, and ARL-ARO QCISS Grant number W911NF-21-1-0005.

\bibliography{main_bib}

\clearpage
\onecolumngrid
\appendix
\section{Appendix: Second moments of $\HTC$ in filled and unfilled sectors}
\label{sec:appendix}
Here, we calculate the second moment of the component of $\HTC(\phi)$ in sector $\H\qj$, as defined in \cref{eq:second_moment_def},
\[
    \Tr\Big\{\big[\HTC^2(\phi)_{q,j}\big]\Big\} \eq \sum_{r=0}^{d_{n}(q,j)-1} \big[\HTC^2(\phi)_{q,j}\big]_{r,r}\,.
\]
Using that the dimension of $\H\qj$ is
\begin{align*}
    d_{n}(q,j) \eq \begin{cases}
    2j+1 & :\,\,q\geq n/2+j\\[4pt]
    q+1+j-\dfrac{n}{2} & :\,\,q< n/2+j\,,
    \end{cases}
\end{align*}
and the relations
\begin{align*}
   r &\eq j+m \qquad\text{and}\qquad k_{q,m} &\eq q-m-\frac{n}{2}\,,
\end{align*}
an explicit calculation yields
\begin{align} \nonumber
    \Tr\Big\{\big[\HTC^2(\phi)_{q,j}\big]\Big\} &\eq \Tr\Big\{\left[\big(e^{i2\phi}(J_+a)^2 + e^{-i2\phi}(J_-a^{\dag})^2 + J_+J_-aa^{\dag} + J_-J_+a^{\dag}a\big)_{q,j}\right]\Big\}
    \\[12pt]\nonumber
    &\hspace{-50pt}\eq \sum_{m=-j}^{d_n(q,j)-1-j}\bra{j,m,\al}\otimes\bra{k_{q,m}}_{\text{osc}}\big(J_+J_-aa^{\dag} + J_-J_+a^{\dag}a\big)\ket{j,m,\al}\otimes\ket{k_{q,m}}_{\text{osc}} \qquad\left(\eq\Tr\Big\{\big[(\HTC^2)_{q,j}\big]\Big\}\right)
    \\[12pt]\nonumber
    &\hspace{-50pt}\eq \begin{cases}\displaystyle
        2\times\sum_{m=-j}^{j}(j+m)(j-m+1)\left(q-m-\frac{n}{2}+1\right)&:\,\,q\geq \dfrac{n}{2}+j\qquad\text{(filled)} \\[16pt]\displaystyle
        2\times\sum_{m=-j}^{q-n/2}(j+m)(j-m+1)\left(q-m-\frac{n}{2}+1\right)
    &:\,\,q<\dfrac{n}{2}+j\qquad\text{(unfilled)}\,.\end{cases}\\[12pt]
    &\hspace{-50pt}\eq \begin{cases}
        \dfrac{2}{3}\times j{(j+1)(2j+1)(2q-n+1)} &:\,\,q\geq \dfrac{n}{2}+j\qquad\text{(filled)} \\[16pt]
        \dfrac{1}{6}\times\left(q-\dfrac{n}{2}+j\right)\left(q-\dfrac{n}{2}+j+1\right)\left(q-\dfrac{n}{2}+j+2\right)\left(3j-q+\dfrac{n}{2}+1\right)&:\,\,q<\dfrac{n}{2}+j
    \qquad\text{(unfilled)}\,.\end{cases}
 \label{eq:second_moments}
\end{align}
For \textit{filled} sectors with fixed $j$, the second moment of $\HTC(\phi)$ is strictly increasing with $q$, so no two filled sectors within the angular momentum $j$ subspace can be unitarily equivalent.

\vspace{3mm}
Furthermore, note from \cref{eq:recall_dimHqj} that  two unfilled sectors $\H_{q_1,j_1}$ and $\H_{q_2,j_2}$ have the same dimension, if and only if, $q_1+j_1=q_2+j_2$, or equivalently, $j_1+m_1=j_2+m_2$.
But then, the first three factors in \cref{eq:second_moments} are equal in the two sectors, so the entire expression is equal only if the last factor is also equal, which implies $3j_1-q_1=3j_2-q_2$.
These two restrictions together imply $j_1=j_2$ and $q_1=q_2$; in other words, $\H_{q_1,j_1}$ and $\H_{q_2,j_2}$ are the same sector.

\end{document}